\documentclass[letterpaper, 10 pt, conference]{IEEEtran}  

\IEEEoverridecommandlockouts                              
\usepackage{graphicx} 
\usepackage{epsfig} 
\usepackage{xcolor}
\usepackage{amsfonts}
\usepackage{latexsym,amssymb,amsmath,amsthm}
\usepackage{tabularx}
\usepackage{bm}
\usepackage{verbatim}
\usepackage{xspace}
\usepackage[pdfstartview=FitH,pdfpagemode=None,backref,colorlinks=true,citecolor=blue,linkcolor=blue]{hyperref}\usepackage{breakurl}
\usepackage{tikz}
\usepackage{pgfplots}
\usepackage{subfig}
\usepackage{algpseudocode}
\usepackage[ruled,linesnumbered]{algorithm2e}

\newcommand{\cG}{\mathcal{G}}

\newcommand{\cE}{\mathcal{E}}

\def\NN{\mathbb{N}}
\def\RR{\mathbb{R}}

\def\rmax{r_\mathrm{max}}

\newtheorem{lem}{Lemma}

\newtheorem{thm}[lem]{Theorem}
\newtheorem{defin}[lem]{Definition}

\title{\LARGE \bf Frequency violations from random disturbances: an MCMC approach}

\author{John Moriarty$^{1}$, Jure Vogrinc$^{1}$ and Alessandro Zocca$^{2}$%
\thanks{$^{1}$John Moriarty and Jure Vogrinc are with the School of Mathematics,  Queen Mary University of London, London E1 4NS, UK. Email: 
        {\tt\small j.moriarty@qmul.ac.uk} and {\tt\small j.vogrinc@qmul.ac.uk}}%
\thanks{$^{2}$Alessandro Zocca is with the California Institute of Technology, Pasadena CA, 91125, USA. Email: 
        {\tt\small azocca@caltech.edu}}%
}

\begin{document}

\maketitle


\begin{abstract}

The frequency stability of power systems is increasingly challenged by various types of disturbances. In particular, the increasing penetration of renewable energy sources is increasing the variability of power generation and at the same time reducing system inertia against disturbances. In this paper we are particularly interested in understanding how rate of change of frequency (RoCoF) violations could arise from unusually large power disturbances.

We devise a novel specialization, named \textit{ghost sampling}, of the Metropolis-Hastings Markov Chain Monte Carlo method that is tailored to efficiently sample rare power disturbances leading to nodal frequency violations. Generating a representative random sample addresses important statistical questions such as ``which generator is most likely to be disconnected due to a RoCoF violation?" or ``what is the probability of having simultaneous RoCoF violations, given that a violation occurs?'' Our method can perform conditional sampling from any joint distribution of power disturbances including, for instance, correlated and non-Gaussian disturbances, features which have both been recently shown to be significant in security analyses. 

\end{abstract}


\section{INTRODUCTION}
\label{sec:intro}

Frequency stability is a prime concern of transmission system operators, as frequencies instabilities may lead to machines desynchronization and trigger large power outages~\cite{Andersson2005}.
Transmission systems are experiencing increased stress and approaching their stability limits due, for example, to the continued connection of power electronics, increased uncertainty, cross-border power flows and the addition of high-voltage direct current (HVDC) links~\cite{Winter2015}. In particular, system inertia is decreasing as synchronous machines are replaced by inverter-connected distributed generation. Low inertia levels, together with the variability of renewable generation, can lead to large swings in the power system frequency~\cite{Ulbig2014}. While promising mitigations exist including participation from loads~\cite{Zhao2014a,Vincent2016}, distributed energy resources~\cite{Guggilam2017} and virtual inertia~\cite{Poolla2017a}, it is increasingly important to also understand the stability of power system frequency under random disturbances to the network's power injections. In the Irish transmission system, for example, the rate of change of frequency (RoCoF) has been identified as the key limit to allowing high real-time penetrations of wind generation \cite{creighton2013increased}. Further it has been shown that frequency fluctuations have a heavy-tailed distribution~\cite{Schafer2018}, making strong deviations more likely than would be expected under, for example, a Gaussian model of fluctuations. Beyond questions of system stability it has also been shown that stochastic disturbances can cause significant resistive power losses to be incurred in stabilising the system frequency~\cite{Tegling2015}.

In this paper we aim to investigate the extent to which infrequently observed, but large, disturbances to nodal power injections can cause an unexpected rate of change in the power system frequency. In particular, following a disturbance to one or more nodal power injections we model the rate of change of frequency (RoCoF) over the primary control timescale. Motivated by stability considerations for RoCoF relays, we develop a simulation technique to investigate possible causes of RoCoF violations. The frequency considered is either the system average frequency or the set of nodal frequencies, and different characterizations of RoCoF violations are explored.

Estimating the probability of rare events in power systems is computationally challenging. Recent work in this area includes~\cite{Owen2017, NestiZoccaZwart2017, Nesti2017} and the present paper is complementary to such studies. Instead we aim to generate a representative sample of disturbances, {\em conditional} on a RoCoF violation occurring. To this end we present {\em ghost sampling}, a specialization of the Metropolis-Hastings Markov Chain Monte Carlo (MCMC) method. Relative to current probabilistic power system reliability analyses, a key advantage of MCMC is to allow the power disturbances to have arbitrary joint distribution. Whereas independent disturbances have always been assumed so far in the literature (see \cite{PaganiniMallada2017}), the use of MCMC allows the disturbance magnitudes to have both correlation and arbitrary marginal probability distributions. Simultaneous nodal power disturbances with a common cause, due for example to line failures or large weather fronts, may thus be modelled in this framework. 

The low probability of RoCoF violations is, however, a challenge to the standard Metropolis-Hastings MCMC algorithm. In common with many other sampling techniques, the latter approach may result in a large part of computational effort being expended in proposing non-violating states. Further, when a violating state is sampled, standard MCMC chains risk becoming `stuck' in its vicinity (see Section~1.11.2 of \cite{Handbook}). The `ghost proposal' described below mitigates these issues by proposing only violating states.

In the nodal analysis we derive expressions for the full set of RoCoFs at time $t \geq 0$. We also obtain expressions for the set of disturbances whose maximum RoCoF exceeds a given acceptable threshold, which may vary per node. From the  representative sample generated we are able to estimate quantities of interest conditional on a violation, in contrast with the average-case analyses common in the literature, see e.g.,~\cite{PaganiniMallada2017,Tegling2015}.

\section{MODEL DESCRIPTION}
\label{sec:model}
A power system described by the graph $G = (\cG,\cE)$ is considered, with nodes (buses) $\cG = \{1, \dots, n\}$ and $m$  edges (transmission lines) $\cE \subseteq \cG \times \cG$. It is assumed that $G$ is a reduced network in which each bus houses a generation unit, since passive loads can be eliminated via Kron reduction~\cite{DB13,DB10}.

Writing $\omega_j$ for the frequency at node $j \in \cG$, the time evolution of nodal frequencies is modelled via linearized dynamics as
\begin{equation}
\label{eq:swing}
	M_j \dot{\omega}_j + D_j \, \omega_j=p^{\mathrm{in}}_j - p^{\mathrm{out}}_j,  \qquad \forall \, j \in \cG,
\end{equation}
where $M_j>0$ is the inertia of the generator at node $j \in \cG$, $D_j >0$ is the damping/droop control coefficient, and $p^{\mathrm{in}}_j$ and $p^{\mathrm{out}}_j$ represent, respectively, the mechanical power injected by the generator at node $j$ and the net electrical power drawn by the network from node $j$; see~\cite{Kundur1994} for more details.

Reactive power injections and reactive power flows are neglected and the standard assumptions of lossless lines, time-invariant identical voltage magnitudes across all nodes and small-signal approximations~\cite{Purchala2005,Wood2014} are made. In view of these assumptions, the so-called DC power flow approximation holds, namely
\begin{equation}
\label{eq:powerfloweq}
	p^{\mathrm{out}}_j = \sum_{i \in \cG} f_{i,j} = \sum_{i \in \cG} B_{i,j} (\theta_i - \theta_j),
\end{equation}
where $f_{i,j}$ describes the power flow on line $e=(i,j) \in \cE$, $B_e=B_{i,j}\geq 0$ is the (effective) susceptance between nodes $i$ and $j$ and $\theta_j$ denotes the phase angle at node $j\in \cG$. Note that an arbitrary but fixed orientation has been chosen for the edges in $\cE$, which is captured by the incidence matrix $C \in \{-1,0,1\}^{n \times m}$ of $G$, that is
\[
	C_{i,e}=
	\begin{cases}
	1 & \text{ if } e=(i,j),\\
	-1 & \text{ if } e=(j,i),\\
	0 & \text{ otherwise.}
	\end{cases}
\]
Denoting by $B \in \mathbb R^{m \times m}$ the diagonal matrix with the susceptances $\{B_e\}_{e=1,\dots,m}$ as diagonal entries, the relation between line flows and phase angles may be rewritten in matrix form as
\[
	f = B C^T \theta,
\]
where $f \in \mathbb R^m$ and $\theta \in \mathbb R^n$ are the vectors of line flows and phase angles, respectively.

We are interested in how, starting from an equilibrium point, the network reacts to a vector $u \in \mathbb R^n$ of
nodal power disturbances. In view of~\eqref{eq:swing} and~\eqref{eq:powerfloweq}, the \textit{deviations} from their nominal values of the nodal frequencies and line power flows are then described by
\begin{subequations}
\begin{align}
	& M_j \dot{\omega}_j = - D_j \, \omega_j + u_j - \sum_{i \,:\, (i,j) \in \cE} f_{i,j},  &\forall \, j \in \cG, \label{eq:swing_1a}\\
	& \dot{f}_{i,j} = B_{i,j} (\omega_i - \omega_j), \quad &\forall \, (i, j) \in \cE, \label{eq:swing_1c}
\end{align}
\end{subequations}
where, with a minor abuse of notation, the variables $\omega$ and $f$ henceforth denote deviations from the corresponding nominal values at equilibrium. This means, in particular, that at equilibrium all variables in equations (3) are equal to 0.

The entries $u_j$, $j \in \cG$ of the random disturbance vector $u \in \mathbb R^n$ are modelled as continuous random variables with joint probability density function $\pi$, so that for any measurable subset $A \subseteq \mathbb R^n$ we have
\begin{equation}
\label{eq:prob}
	\mathbb{P}[u\in A]=\int_A \pi(u_1,\dots u_n) \, du_1\dots du_n.
\end{equation}
The correlation in renewable generation, alongside correlation in other factors such as loads, has been shown to have a significant effect on power system risk assessment~\cite{li2015transmission}. One advantage of our approach is that the random disturbances $u_j$ are not required to be independent. This is because we aim to simulate typical disturbances $u$ causing frequency violations rather than, for example, to derive closed form expressions for synchronization performance as in~\cite{PaganiniMallada2017} or~\cite{Tegling2015}. 

Further, the errors in renewable power forecasts have been shown to have significantly non-Gaussian distributions. For example, fat tails have been demonstrated in wind power forecast errors~\cite{bludszuweit2008statistical}. This is another advantage of our conditional simulation procedure, since the random disturbances can have a general joint probability density. To illustrate this point the case study presented later in Section~\ref{sec:cs} uses a mixture of uncorrelated Gaussian and correlated, fat tailed non-Gaussian distributions.

The $u_j$ are modelled as step disturbances, namely
\[
	u(t)=u \, \mathbf{1}_{\{t \geq 0\}}.
\]
Thus time $t=0$ is the moment just after the random disturbance(s). The desynchronization effect of $u$ on the frequencies at all nodes $j \in \cG$ will be modelled from time $t=0$ until time $t=\epsilon>0$. This step model is valid when the 
disturbances represented by the $u_j$ can be reasonably approximated as constant over the time interval $[0,\epsilon]$. (In the case study below we take $\epsilon=0.5\mathrm{s}$.)

Our method in the remainder of the paper has two parts, as follows:
\begin{enumerate}
\item characterise the `safe region' $K \subset \RR^n$ of disturbances $u \in \RR^n$ which do not give rise to frequency violations;
\item generate a statistically representative sample from its complement $K^c$. 
\end{enumerate}

Frequency violations will be characterised using the RoCoF, by which we mean $|\dot\omega|$, the magnitude of the rate of change of frequency. The analysis begins with the rate of change of the system frequency, before moving to the consideration of nodal frequencies. The latter context is particularly pertinent when nodal frequencies are considered, since generating machines are protected by RoCoF-sensitive relays which observe only the local nodal frequency. 


The rest of the paper is organized as follows. In Section~\ref{sec:sfreq} we introduce step 1) in the simpler context of analysing the system frequency. In Section~\ref{sec:nodal} the nodal frequency dynamics are first established, and the step 1) is then carried out in this context. Step 2) is developed in Section~\ref{sec:math}, and an illustrative case study for our analysis is provided in Section~\ref{sec:cs}.

\section{System frequency}\label{sec:sfreq}

The \textit{system frequency} or \textit{center of inertia} (COI) is defined as (see, for example, \cite{Ulbig2014}):
\[
	\bar{\omega}(t):=\frac{\sum_{i \in \cG} M_i \, \omega_i(t)}{\sum_{i \in \cG} M_i}.
\]
This model has been studied in \cite{PaganiniMallada2017} under the simplifying condition that there exist rating parameters $f_1,\dots,f_n >0$ with $\max_i f_i=1$ such that the inertia and damping coefficients of generator $i$ are given respectively by
\begin{equation}
\label{eq:condition1}
	M_i= f_i M \quad \text{ and } \quad D_i = f_i D, \quad i \in \cG,
\end{equation}
where $M$ and $D$ are those of the machine $j$ such that $f_j=1$.
In particular, it is shown in the latter paper (cf. Eq.(18)) that the following holds under condition~\eqref{eq:condition1}:
\begin{equation} \label{eq:systemfreq}
	\bar{\omega}(t)= g(t) \sum_{i} u_i, \quad t >0,
\end{equation}
where $g(t) := \left (\sum_{i} D_i \right )^{-1} \left( 1 - e^{-(D/M) t} \right)$. 

The swing and network dynamics in~\eqref{eq:swing_1a}-\eqref{eq:swing_1c} can be enriched to incorporate the turbine control dynamics, yielding the following third-order model
\begin{subequations}
\begin{align}
	& \dot{\omega}_j = -\frac{1}{M_j} \Big ( D_j \, \omega_j - q_j - u_j +  \sum_{i \,:\, (i,j) \in \cE} f_{i,j} \Big) \label{eq:turb1}\\
	& q_j = - \frac{1}{\tau} ( R_j^{-1} \omega_j + q_j)\label{eq:turb2}
\end{align}
\end{subequations}
where $q_j$ is the (variation of) turbine power, $R_j$ the droop coefficient, and $\tau$ is the turbine time constant (which is uniform across different generators). 
If in addition to~\eqref{eq:condition1} we further assume that $R_j^{-1}= f_j R^{-1}$ for every $j \in \cG$ and that the system is under-damped, i.e.,~$\omega_d^2:= \frac{D+R^{-1}}{M \tau} - \frac{1}{4} \left ( \frac{1}{\tau} + \frac{D}{M} \right)^2 >0$, the system frequency still obeys an equation of the same form as~\eqref{eq:systemfreq}, where the function $g(t)$ is now a more involved function, namely
\[
	g(t) := \frac{1- e^{ - \eta t} (\cos(\omega_d t) - \frac{\gamma-\eta}{\omega_d} \sin (\omega_d t) )}{(\sum_i D_i+R_i^{-1})(D+R^{-1})},
\]
where $\eta:=\frac{1}{2} \Big( \frac{1}{\tau} + \frac{D}{M} \Big)$ and $\gamma:=\Big( \frac{1}{\tau} - \frac{R^{-1}}{M} \Big)$.

From both models, with or without turbine control, it is proved in~\cite{PaganiniMallada2017} that the maximum RoCoF occurs at time $t \downarrow 0+$ and is equal to
\[
	\max_{t > 0} \left|\frac{d}{d t} \overline{\omega}(t)\right| = \lim_{t \to 0+} \left|\frac{d}{d t} \bar{\omega}(t)\right| = \Big |\sum_{i \in \cG}u_i\Big | \left | \dot{g}(0)\right |.
\]



Hence, for both dynamics, with or without turbine control, 
the set of disturbances $u \in \mathbb R^n$ whose maximum induced RoCoF 
 does not exceed a predetermined threshold $\rmax$ is simply
\[
	K_{\mathrm{MS}} =\left\{ u \in \mathbb R^n : \Big|\sum_i u_i\Big| \leq \frac{\rmax}{|\dot g(0)|} \right\}.
\]
The region $K_{MS}$ is a convex polytope in $\mathbb R^n$ -- that is, the intersection of a number of half-spaces, which is not necessarily bounded. (Here the subscript M refers to the Maximum RoCoF metric and S to the System frequency).


The metric of average absolute RoCoF may alternatively be considered.
Over the time interval of length $\epsilon$ following the disturbance $u$ this is given by $\Omega(\epsilon)$, where
\begin{equation}\label{eq:avg}
\Omega(\epsilon) := \frac 1 \epsilon \int_0^\epsilon |\dot{\bar{\omega}}(t)|dt = \frac 1 \epsilon \left|\sum_i u_i\right| \int_0^\epsilon |\dot g(t)|dt.
\end{equation}
The set of disturbances whose average induced RoCoF over $t \in [0,\epsilon]$ does not exceed the threshold $\rmax$ is thus
\[
	K_{\mathrm{AS}} = \left\{ u \in \mathbb R^n : \Big|\sum_i u_i\Big| \leq \frac{\epsilon \, \rmax}{\int_0^\epsilon |\dot g(t)|dt} \right\}.
\]
Since the turning points of $g$ can be calculated analytically for both models, the evaluation of the integral is straightforward. The set $K_{\mathrm{AS}}$ of disturbances is also a convex polytope. 

In summary, the `safe region' of disturbances $u \in \RR^n$ inducing a rate of change in the system frequency less than a threshold $r$ is a convex polytope, for either the maximum or average absolute RoCoF metric, and with or without turbine control. In the next section we are able to establish a similar result for the set of nodal frequencies under swing dynamics.


\begin{figure}[htbp]
\centering
\subfloat[Time evolution over the first $2\mathrm{s}$ of the nodal frequency deviations $\omega_1(t)$, $\omega_2(t)$, and $\omega_3(t)$ and the system frequency $\bar{\omega}(t)$ ($\mathrm{Hz}$) in the case study of Section \ref{sec:cs} after a random disturbance.]{\includegraphics[width=0.45\textwidth]{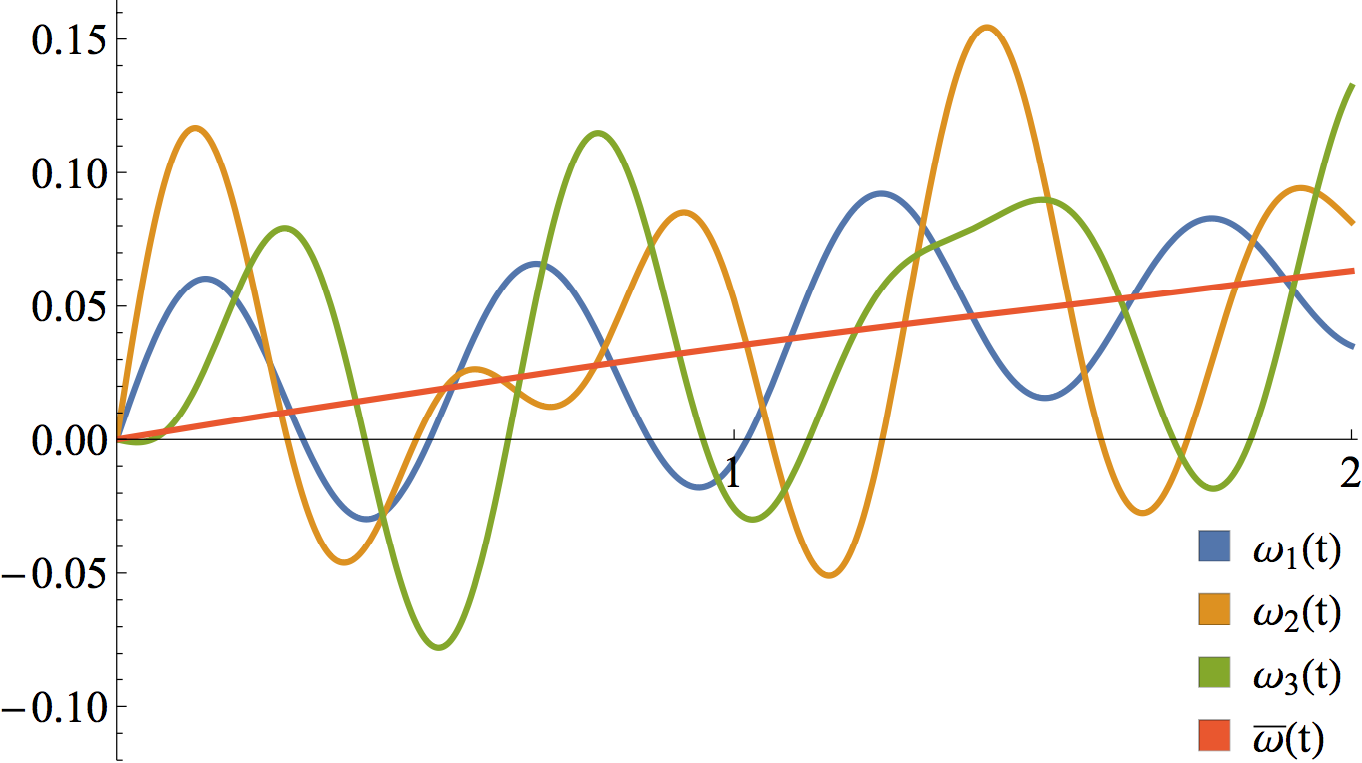}}\\
\subfloat[Corresponding evolution over the first $2\mathrm{s}$ of frequency speed deviations $\dot{\omega}_1(t)$, $\dot{\omega}_2(t)$, $\dot{\omega}_3(t)$, and $\dot{\bar{\omega}}(t)$ ($\mathrm{Hz}/\mathrm{s}$) for the same random disturbance as in Fig.~\ref{fig:nodalfrequencies}(a). The dashed horizontal lines represent the RoCoF threshold $\rmax=1\mathrm{Hz}/\mathrm{s}$.]{\includegraphics[width=0.45\textwidth]{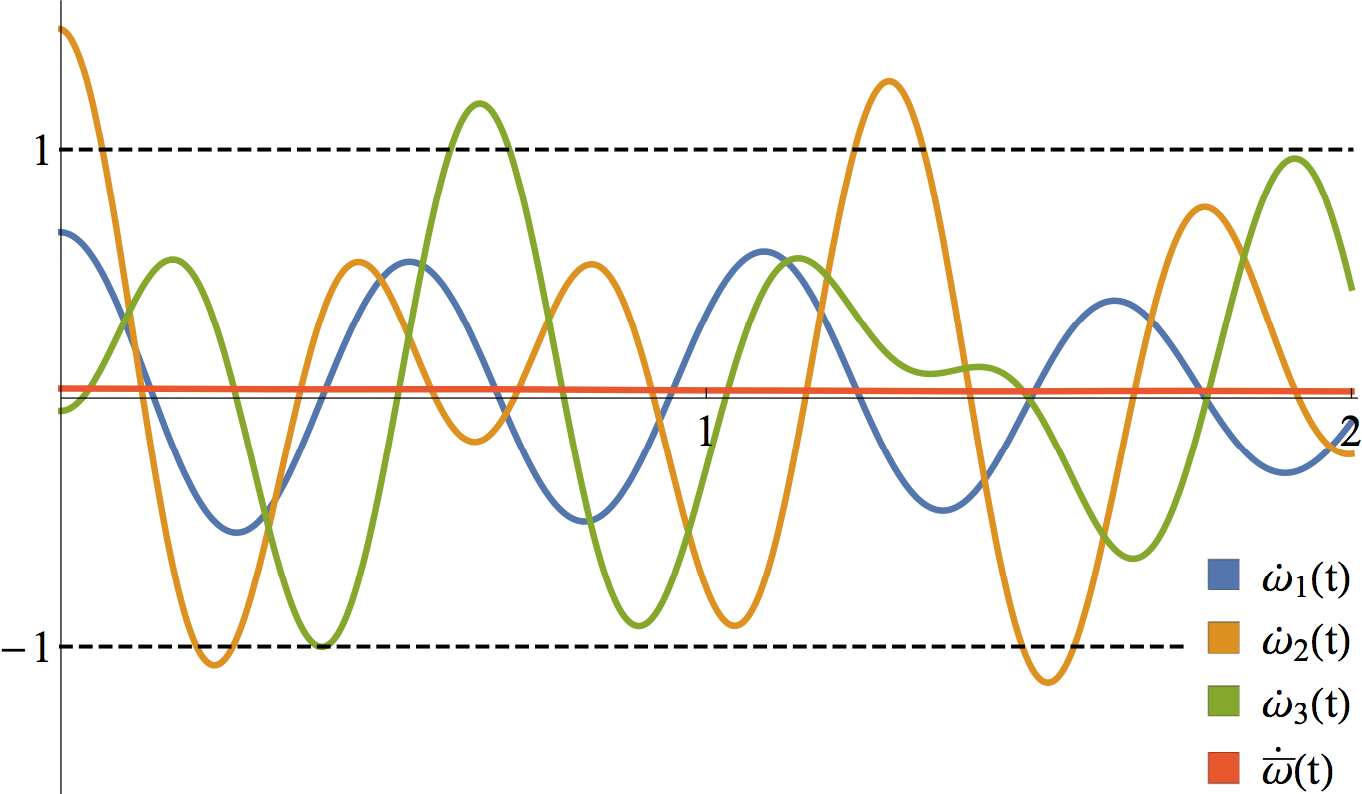}}
\caption{Post-disturbance traces of some nodal frequencies for the case study about the IEEE 39 New England interconnection system presented of Section~\ref{sec:cs}.}
\label{fig:nodalfrequencies}
\end{figure}

\section{RoCoF violations for nodal frequencies}
\label{sec:nodal}

While the system frequency is of inherent interest, it cannot capture de-synchronization within the system. To illustrate this point, Fig.~\ref{fig:nodalfrequencies} plots three nodal frequency traces together with the system frequency following a random disturbance. These traces are simulated from the system described in the case study of Section \ref{sec:cs}. It is clear from the figure that, under either the maximum or average absolute RoCoF metric, it is possible for a given threshold to be simultaneously respected by the system frequency and violated by one or more nodal frequencies. Further, as mentioned above, from the practical perspective it is nodal frequencies which trigger the operation of RoCoF-sensitive generator protection relays which can, in turn, lead to significantly more serious frequency violations. The main aim of this paper is therefore sampling the typical causes of nodal frequency violations, and we focus on this from now on. We remark that the generator inertia and damping coefficients $M_i$ and $D_i$ respectively may be arbitrary, so that our setting is more general than that presented in Section~\ref{sec:sfreq} (cf.~condition~\eqref{eq:condition1}).

Let $M \in \RR^{n \times n}$ and $D \in \RR^{n \times n}$ be the diagonal matrices containing the generator inertias and damping factors respectively, and $u \in \RR^n$
the random vector of disturbances. Together with the notation from the previous section, the swing equations~\eqref{eq:swing_1a} and~\eqref{eq:swing_1c} read
\begin{eqnarray*}
\begin{bmatrix}
\dot{\omega}\\
\dot{f}
\end{bmatrix}&=&
\begin{bmatrix}
-M^{-1} D & -M^{-1} C\\
B C^T & \mathbb{O} 
\end{bmatrix}\cdot
\begin{bmatrix}
\omega\\
f
\end{bmatrix}
+
\begin{bmatrix}
M^{-1}\\
0
\end{bmatrix}
u,
\\
\begin{bmatrix}
\omega(0)\\
f(0)
\end{bmatrix}&=&
\begin{bmatrix}
0\\
0
\end{bmatrix}.
\end{eqnarray*}
By differentiation we obtain
\begin{eqnarray*}
	\ddot{\omega}&=&-M^{-1} D \,  \dot{ \omega} -M^{-1} C B C^T \, \omega
	\\
	&=& -M^{-1} D \, \dot{ \omega} -M^{-1} L \, \omega, \\
	\dot{\omega}(0)&=& M^{-1} u,
\end{eqnarray*}
where $L:=C B C^T$ is the weighted Laplacian matrix of the graph $G$.
Ignoring the line flows, we thus obtain a homogeneous dynamical system of the form $\dot{x} = A x$, with 
\begin{eqnarray*}
x &=& 
\begin{bmatrix}
\dot{\omega}\\
{\omega}
\end{bmatrix},
\quad
A=
\begin{bmatrix}
-M^{-1} D & -M^{-1} L\\
I & \mathbb{O} 
\end{bmatrix},
\\
x(0)&=&
\begin{bmatrix}
M^{-1} u\\
0
\end{bmatrix},
\end{eqnarray*}
whose solution is 
\begin{equation}
\label{eq:expsol}
	x(t) = \exp (t A) \, x(0).
\end{equation}


Henceforth the maximum RoCoF metric will be used to characterise frequency violations: the average metric may be applied in a similar manner but we reserve this for future work. The set of disturbances $u$ which do not induce RoCoF violations will again be referred to as the `safe region'. 



As confirmed by Fig.~\ref{fig:nodalfrequencies}(b), and in contrast to the system frequency models of Section~\ref{sec:sfreq}, for a fixed node $j$ the maximum RoCoF $\dot{\omega}_j$ does not in general occur at time $0$. Let us therefore consider sampling $\dot{\omega}_j$ at times $\frac{n}{N}\epsilon$, $n=0, \ldots, N$. Although in principle this involves no loss of generality since digital RoCoF measurements have a discrete sampling rate, we note that any lower sampling rate $N/\epsilon$ should be chosen carefully to avoid an excessive loss of sensitivity (a sensitivity analysis for $N$ is provided in the case study of Section \ref{sec:cs}). Define the `node-$j$ safe region' $K^{(j,N)}$ by
\begin{align}
	K^{(j,N)} &= \bigcap_{n=0}^N K^{(j,N)}_n, \text{ where } \\
	K^{(j,N)}_n &= \left\{ u \in \mathbb R^n : \left|\dot{\omega}_j\left(\frac{n}{N}\epsilon\right)\right| \leq \rmax \right\}.
	\label{eq:knjn}
\end{align}

It follows from \eqref{eq:expsol} that $K^{(j,N)}_n$ is given by the convex polytope
\begin{align}
	K^{(j,N)}_n &= \left\{ u \in \mathbb R^n : \left| \exp \Big ( \frac{n \, \epsilon}{N} A \Big )_j
	\begin{bmatrix}
		M^{-1} u\\
		0
	\end{bmatrix}	
	\right| \leq \rmax\right\},\nonumber
\end{align}
where $\exp (t A)_j$ denotes the $j$-th row of the matrix $\exp( t A)$. Hence $K^{(j,N)}$ and the `all-nodes safe region' $K^{(N)}$ are also convex polytopes, where
\begin{align}\label{eq:Kn}
K^{(N)} = \bigcap_{j \in \cG} K^{(j,N)}. 
\end{align}
(Note that, clearly, different thresholds $\rmax$ could be chosen per node to enable modelling of differing protection relay settings for differing types of generating machine, or to enable to modelling of DC links, and the safe region would again be a convex polytope).

Having characterised the safe region, we now turn to the problem of generating a representative sample from its complement. In the next section we describe how the Metropolis-Hastings MCMC algorithm, a commonly used technique for generating random samples, may be efficiently adapted for this purpose.

\section{Ghost sampling}
\label{sec:math}

Recalling from~\eqref{eq:prob} that the entries $u_j$ of the random disturbance $u$ are modelled as continuous random variables with a joint probability density function $\pi$, 
%
the goal in this section is to sample efficiently from the conditional joint density, or {\em target}, $$\frac{\pi(u)\mathbf{1}_{K^c}(u)}{\pi(K^c)},\quad \text{where} \quad \pi(K^c)=\int_{K^c}\pi(v)dv.$$

Since MCMC sampling methods do not require strong assumptions on the target density they are ideally suited to such  problems \cite{Handbook}, \cite{tierney}. However the event $K^c$ is in principle rare, which may cause problems of computational inefficiency. Below we describe the \emph{ghost sampler}, a particular Metropolis-Hastings (MH) algorithm  
designed to be efficient in this context.

The ghost sampler is defined in Section \ref{sec:ghost_sampler}, and in Section \ref{sec:gsprop} it is shown that the generated samples may be used to approximate statistics of the corresponding target distribution. We will show in the case study of Section \ref{sec:cs} that this enables important statistical questions to be addressed such as ``which generator is most likely to be disconnected due to a RoCoF violation?" or ``what is the probability of two simultaneous RoCoF violations being caused, given that a violation occurs?''.





\subsection{Ghost sampling algorithm}\label{sec:ghost_sampler}

Ghost sampling lies in the class of Metropolis-Hastings algorithms \cite{roberts}, \cite{tierney}. That is, beginning at $X_0\in\RR^n$, for every $i=1,2,\dots$ we generate a proposal $Y_{i+1}$ distributed according to a density $q(X_i,y)dy$, and evaluate the acceptance probability 
\begin{equation}\label{eq:ap}
	\small \alpha(X_i,Y_{i+1})=\min\left(1, \frac{\pi(Y_{i+1})\mathbf{1}_{K^c}(Y_{i+1})q(Y_{i+1},X_i)}{\pi(X_i)\mathbf{1}_{K^c}(X_i)q(X_i,Y_{i+1})}\right),
\end{equation}
which is interpreted as one if $\pi(X_i) \mathbf{1}_{K^c}(X_i)q(X_i,Y_{i+1})=0$. With probability $\alpha(X_i,Y_{i+1})$ the proposal is accepted and we set $X_{i+1}=Y_{i+1}$, otherwise it is rejected and $X_{i+1}=X_i$. The aim is to generate a Markov chain $X_1, X_2\dots$ with stationary distribution equal to $\frac{\pi \mathbf{1}_{K^c}}{\pi(K^c)}$ which satisfies the law of large numbers (LLN), meaning that sample averages $\frac 1 n\sum_{i=1}^n f(X_i)$ for large $n$ provide good estimates for the actual conditional expectations $$\frac{\int_{K^c}f(v)\pi(v)dv}{\pi(K^c)}=\mathbb{E}_\pi[f(X)|X\notin K].$$ 

Commonly there is an underlying symmetric density $q\colon \RR^n\to\RR$ (that is, with $q(x)=q(-x)$) and the proposal density used in MH algorithm is $q(x,y)=q(|y-x|)$ where we abuse notation slightly by denoting both with $q$. In this case the algorithm is called \emph{Symmetric Random Walk Metropolis} (SRWM), $q(x,y)=q(y,x)$ holds and the $q$ terms in \eqref{eq:ap} cancel out. Typical examples are $Y_{i+1}\sim N(x,\sigma^2I_n)$ or $Y_{i+1}\sim X_I+U([-\delta,\delta]^d)$, that is, the proposal is drawn from a normal (resp. uniform) distribution centred at $X_i$. Note from \eqref{eq:ap} that knowledge of $\pi \mathbf{1}_{K^c}$ suffices and the normalising constant $\pi(K^c)$ need not be known.

A well-known difficulty with the MH algorithm arises when the target density is multi-modal (see Section~1.11.2 in \cite{Handbook}). In the present application to rare event sampling, where the ``common" events are removed from $\pi$, we may be left with a target density $\frac{\pi \mathbf{1}_{K^c}}{\pi(K^c)}$ with multiple, well-separated local modes. The difficulty arises since a large proportion of the proposals $X_i$ will lie in $K$ and thus be rejected (since then $\alpha(X_i,Y_{i+1})=0$), rendering the method inefficient.

The \emph{ghost proposal} is designed to circumvent these issues by moving `through' $K$, and is now described in the case when $K$ is a bounded convex polytope (clearly $K$ should also have nonzero volume).

Fix an SRWM algorithm with proposal density $q$ and target $\frac{\pi \mathbf{1}_{K^c}}{\pi(K^c)}$. Denote the boundary of $K$ by $\delta K$, and let the current state of the chain be $X_i\notin K$. First, generate a SRWM proposal $Y_{i+1}$ and denote $\varphi_i:=\frac{Y_{i+1}-X_i}{|Y_{i+1}-X_i|}$. Then with probability 1 we have $Y_{i+1} \neq X_i$ and the ray from $X_i$ passing through $Y_{i+1}$ intersects $\delta K$ either twice (cf.~Fig.~\ref{fig:ghost}(b)-(c)) or not at all (cf.~Fig.~\ref{fig:ghost}(a)). If there are two numbers $t_2>t_1>0$ such that $X_i+t \varphi_i \in \delta K$, modify the proposal to $Z_{i+1}=Y_{i+1}+(t_2-t_1)\varphi$ (cf.~Fig.~\ref{fig:ghost}(d)). If there are no such points, set $Z_{i+1}=Y_{i+1}$ (cf.~Fig.~\ref{fig:ghost}(a)).  
Finally perform a MH step, accepting the proposal $Z_{i+1}$ with probability $\alpha(X_i,Z_{i+1})$ given by~\eqref{eq:ap} and setting $X_{i+1}=Z_{i+1}$, otherwise rejecting the proposal and setting $X_{i+1}=X_i$. This procedure is depicted in Fig.~\ref{fig:ghost}  and pseudocode is provided in Algorithm~\ref{alg:ghost}.

\begin{figure}[!h]
\centering
\vspace{-0.1cm}
\subfloat[]{\hspace{0cm}\includegraphics[scale=0.64]{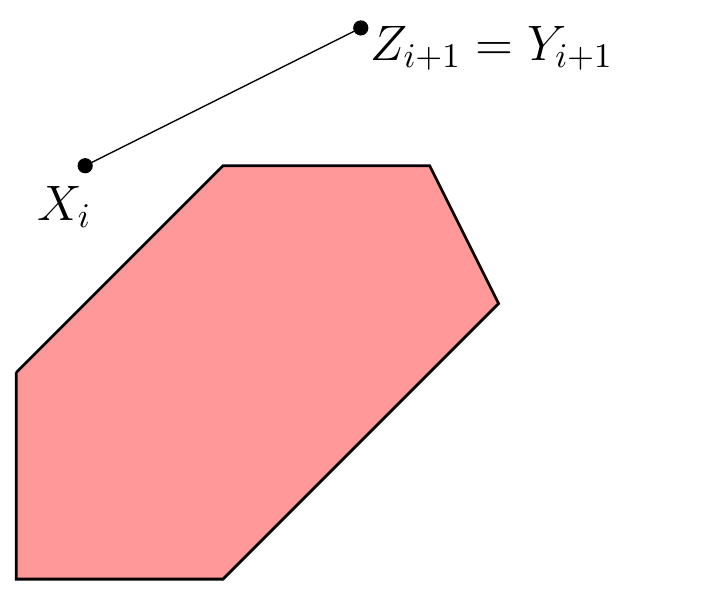}}
\subfloat[]{\hspace{-0.35cm}\includegraphics[scale=0.64]{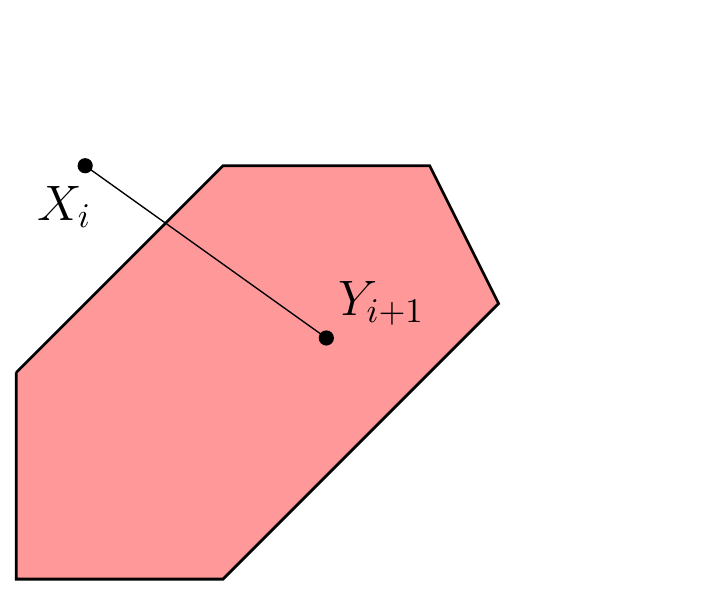}}\\
\vspace{-0.1cm}
\subfloat[]{\hspace{0cm}\includegraphics[scale=0.64]{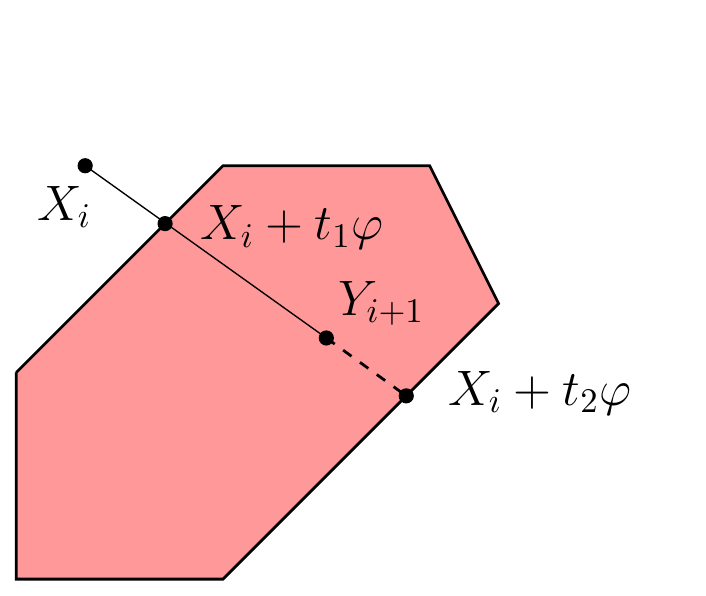}}
\subfloat[]{\hspace{-0.35cm}\includegraphics[scale=0.64]{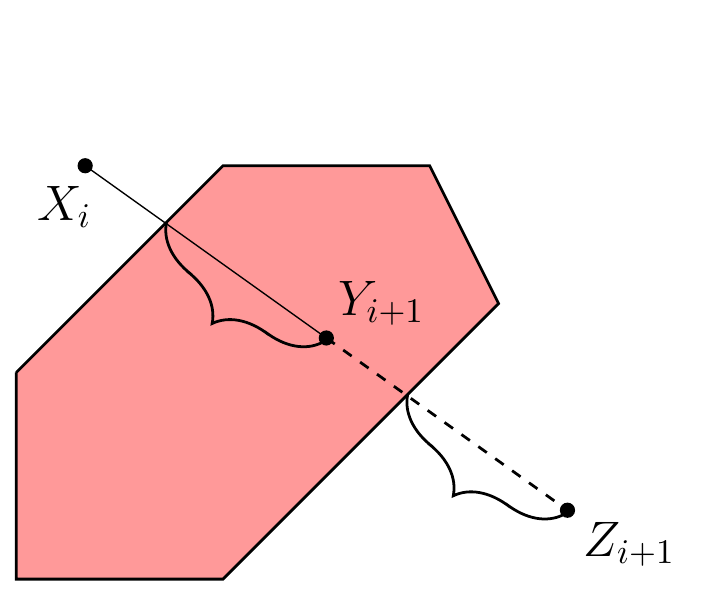}}
\caption{Illustration of the key ideas underlying the ghost sampling method.}\label{fig:ghost}
\end{figure}
\vspace{-0.1cm}
\begin{algorithm}[!h]
\SetAlgoLined
\SetKwInOut{Input}{input}\SetKwInOut{Output}{output}
\Input{$X_i \in K^c$}
\BlankLine
\DontPrintSemicolon
Generate SRWM proposal $Y_{i+1}$;\\
Calculate direction $\varphi_i=\frac{Y_{i+1}-X_i}{|Y_{i+1}-X_i|}$;\\
Calculate all points $T:=\{ t >0 \,:\, X_i+t \varphi \in \delta K\}$;\\
\uIf{$T=\{t_1,t_2\}$}{$Z_{i+1}=Y_{i+1}+(t_2-t_1)\varphi$;}
\Else{$Z_{i+1}=Y_{i+1}$;}
Generate a uniform random variable $U$ on $[0,1]$;\\
\uIf{$U \leq \alpha(X_i,Z_{i+1})$}{$X_{i+1}=Z_{i+1}$;}
\Else{$X_{i+1}=X_i$;}
\KwRet{$X_{i+1}$}\;
\caption{Ghost sampler ($i$-th step)}\label{alg:ghost}
\end{algorithm}

The following example illustrates how ghost sampling improves upon the standard MH algorithm in the present context. Set $K=\{(x,y)\in\RR^2: |x|+|y|<7\}$, let $\pi$ be the two-dimensional Gaussian density with zero mean and covariance matrix $\text{diag}(4,1)$ and let $q$ be the density of a standard two-dimensional Gaussian random variable.

A MH chain starting to the left of the diamond-shaped set $K$ (cf. Fig.~\ref{fig:diamond}) will have difficulties crossing to the right side of the diamond, since a direct move to the other side is unlikely and any sequence of steps towards the right side is likely to suffer rejections because the values of $\pi$ are much smaller around the top and bottom vertices of $K$ then around the left and right vertices. The ghost sampler, however, is likely to make a direct move between the left and right sides. This can be seen in Fig.~\ref{fig:diamond}, where scatter plots are given of values taken by the standard MH chain (in red) and the ghost sampler (blue). In this way the ghost sampler is designed to more efficiently explore the rare event $K^c$.

\begin{figure}[!h]
\centering
\vspace{0.5cm}
\includegraphics[width=0.48\textwidth]{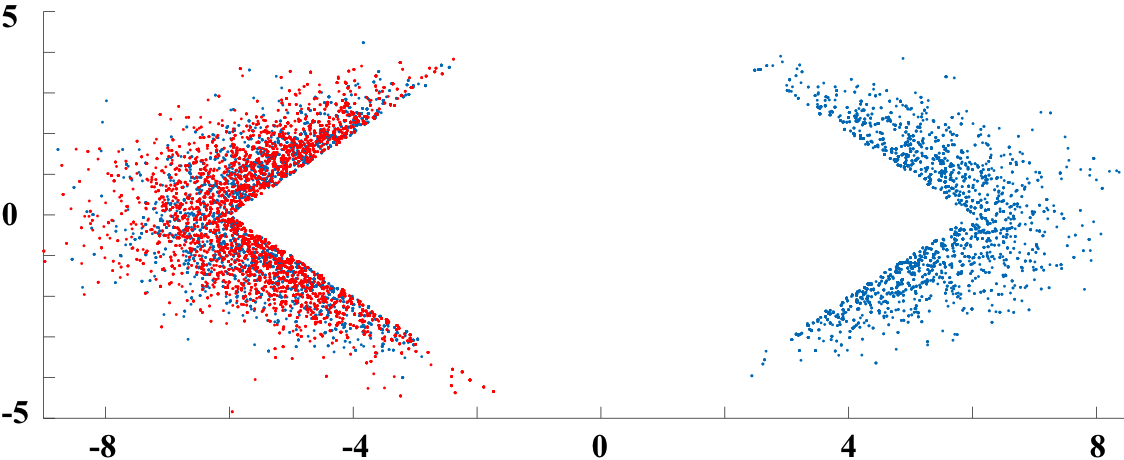}
\caption{Example values taken by the ghost sampler (blue) and standard MH chain (red).}\label{fig:diamond}
\end{figure}

\subsection{Ghost sampler properties}\label{sec:gsprop}
We now state necessary mathematical properties of the ghost proposal in a more general setting, which are proved in the Appendix. More precisely, we assume only that $K$ is closed and {\em ray-bounded} (see Definition \ref{def:man}).

For a starting point $x\in K^c$, a direction $\varphi$ lying in the unit sphere $\mathbb{S}^{n-1}$ and a distance $r>0$, denote the total length between $x$ and $x+r\varphi$ that lies outside $K$ by 
\begin{equation}
l^\varphi_x(r):=\int_0^r\mathbf{1}_{K^c}(x+t\varphi)dt.
\end{equation}
Also denote the mapping $T^K_x\colon \RR^n\to\RR^n$ by
\begin{equation}
T^K_x(x+r\varphi):=x+l^\varphi_x(r)\varphi
\end{equation}
and a modified proposal density termed the \emph{ghost density} as
\begin{equation}
\label{eq:ghostdens}
	q_K(x,x+r\varphi):=q(l^{\varphi}_x(r)\varphi)\left(\frac{l^\varphi_x(r)}{r}\right)^{n-1}\mathbf{1}_{K^c}(x+r\varphi),
\end{equation}
where $l^\varphi_x(r)/r$ is interpreted as $1$ if $r=0$.

Note that for $x,y\in K^c$ we have $q_K(x,y)=q_K(y,x)$, since
\[
	l^{\frac{y-x}{|y-x|}}_x(|y-x|)=l^{\frac{x-y}{|x-y|}}_y(|x-y|).
\]
Intuitively the map $T^K_x$ contracts each ray emanating from $x$ by removing its intersection with $K$. It is the reverse of the ghost sampling modification, in the sense that $T^K_{X_i}(Z_{i+1})=Y_{i+1}$.
\begin{lem}
If $K$ is a closed set and $x\in K^c$, then the map $T^K_x\colon K^c\to \RR^n$ is injective.
\end{lem}

\begin{defin}\label{def:man}
A closed set $K$ is said to be \emph{ray-bounded} if the map $T^K_x: K^c \to \RR^n$ is surjective for all $x\in K^c$.
\end{defin}
\noindent In the context of convex polytopes ray boundedness simply means that any ray starting outside $K$ and intersecting it will also exit it. For instance, the set $\{(x,y)\in\RR^2: |x|<1\}\subset\RR^2$ is ray-bounded but not bounded.

For each $x \in K^c$ and measurable set $A \subset \RR^n$ we define
\begin{eqnarray*}
	Q(x,A)&:=&\int_{\RR^n}q(x,y)\, \mathbf{1}_A(y) \, dy, \\
	Q_K(x,A)&:=&\int_{\RR^n}q_K(x,y) \, \mathbf{1}_A(y) \, dy.
\end{eqnarray*}

The following lemma shows that in the case of ray-bounded polytope procedure described in Section~\ref{sec:ghost_sampler} indeed results in the ghost density defined in \eqref{eq:ghostdens}.

\begin{lem}\label{lemma:Q}
For every $x\in K^c$ and every measurable set $A\subset \RR^n$ we have
\[
	Q_K(x,A)=Q(x,T^K_x(A\cap K^c)).
\]
\end{lem}

Whenever $x\in K^c$ and $K$ is ray-bounded, Lemma \ref{lemma:Q} gives that $Q_K(x,K^c)=Q(x,T^K_x(K^c))=Q(x,\RR^n)=1$, which means that $Q_K$ is a \textit{Markov kernel} on $K^c$. It then follows from~\cite[Section~2.3.1]{tierney} that the measure with density $\frac{\pi \mathbf{1}_{K^c}}{\pi(K^c)}$ is the unique stationary probability measure of the MH algorithm with the ghost proposal density $q_K$.

To complete the justification of ghost sampling, we are also able to show the LLN:

\begin{thm}\label{thm:LLN}
Let $K$ be closed and ray-bounded. Let $X_1,X_2,\dots$ be a Markov chain generated by the ghost proposal $q_K$ derived from a SRWM proposal with density $q$ which is strictly positive on $\RR^n$. Then the strong law of large numbers holds, that is, for every $\pi$-integrable function $f$ we have:
$$\frac{1}{n}\sum_{i=1}^n f(X_i)\xrightarrow[\text{a.s.}]{n\to\infty}\mathbb{E}_\pi[f(X)|X\notin K].$$
\end{thm}

\section{Case study: IEEE 39 New England network}\label{sec:cs}
In this section we illustrate how the ghost sampler enables inference about RoCoF violations.
As a case study we consider the IEEE 39-bus New England interconnection system, which has 10 generators and 29 load nodes, see Fig.~\ref{fig:ieee39}(a). The system parameters for our experiments are taken from the Matpower Simulation Package~\cite{zimmerman2011matpower}. 

We consider the Kron-reduced version of the aforementioned system, which is illustrated in Fig.~\ref{fig:ieee39}(b). The thickness of the edges in Fig.~\ref{fig:ieee39}(b) is proportional to the equivalent susceptance between the two corresponding generator nodes.
\begin{figure}[!h]
\centering
\vspace{-0.3cm}
\subfloat[]{\includegraphics[scale=0.195]{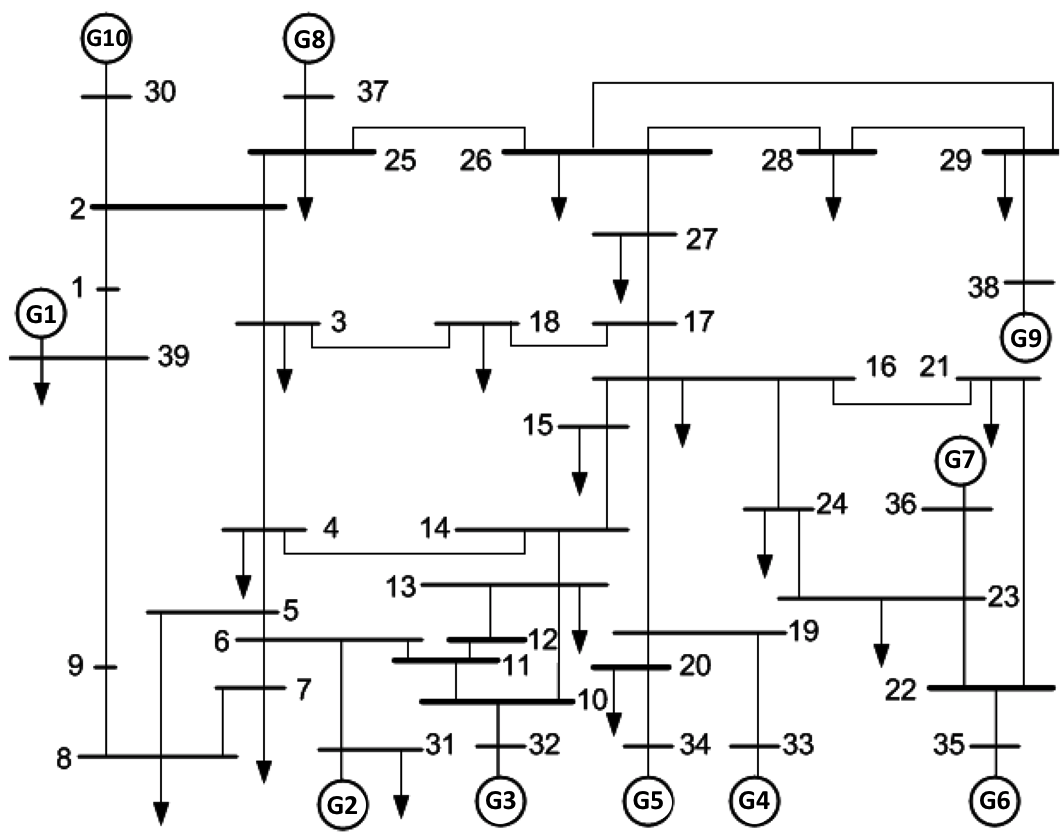}}\\
\subfloat[]{\includegraphics[scale=0.195]{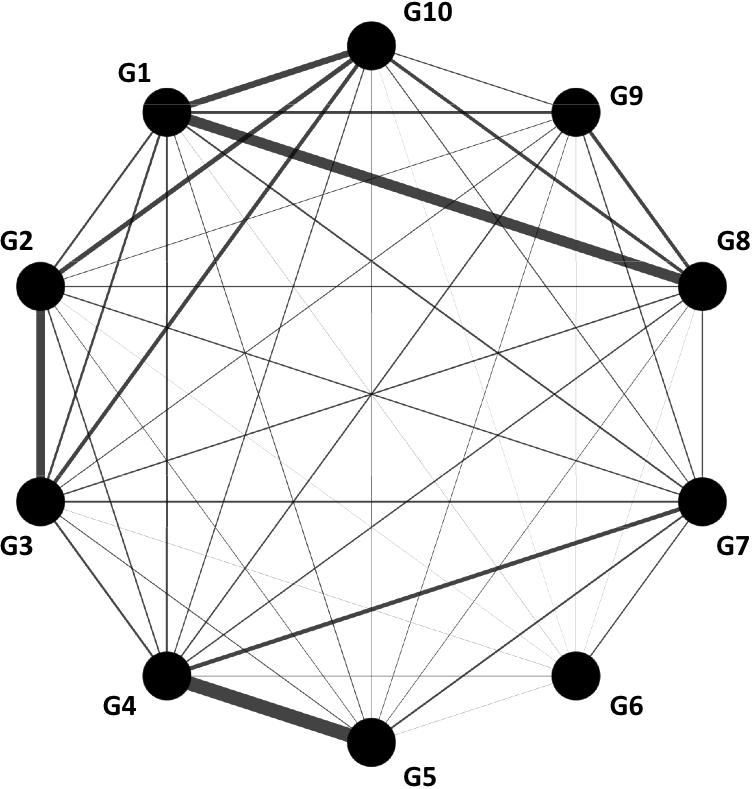}}
\caption{(a) Line diagram of the IEEE 39-bus system and (b) the Kron-reduced version of the IEEE 39-bus system with only the 10 generator nodes.}
\label{fig:ieee39}
\end{figure}

The ghost sampler is capable of sampling from any continuous distribution one may want to consider for the power disturbances, in particular those featuring heavy-tailed or correlated components. Aiming to illustrate its potential, we thus consider a mixed distribution that prescribes the disturbances $u_1$ and $u_2$ in generators $1$ and $2$ to be correlated and heavy-tailed, while the remaining generators are assumed to have i.i.d.~Gaussian disturbances. 

More specifically, we model the disturbances $u_3,\dots, u_{10}$ as independent Gaussian random variables with zero mean and standard deviations of the size $1/65$ times the nominal power injections of the associated generators, while the disturbances $u_1$ and $u_2$ are modelled by a correlated heavy-tailed density
\begin{align}
(u_1,u_2)\sim\frac{1}{1+(30(u_1-u_2/2))^4}\cdot \frac{1}{1+(30(u_2-u_1/2))^4}.\nonumber
\end{align}

RoCoF violations are characterised using the maximum RoCoF metric with threshold $\rmax=1\mathrm{Hz}/\mathrm{s}$, which corresponds to the safe region $K^{(N)}$ introduced in the previous section, see~\eqref{eq:Kn}. The duration considered is $\epsilon=0.5\mathrm{s}$ and the sensitivity of the results is examined with respect to $N$, taking $N=1,5,20,50,100$. The ghost sampler uses a Gaussian proposal $N(0,\sigma^2I)$, whose standard deviation $\sigma^2=10^{-3}$ has been tuned so that approximately 15\% of the proposed moves are accepted, as suggested in~\cite{NealRoberts}.

For each value of $N$, $10^6$ disturbances $u$ from the complement of $K^{(N)}$ were sampled after discarding an initial burn-in period. 
For each generator, Table~\ref{table1} reports its probability of disconnection due to a nodal RoCoF violation. Note that generator $10$ was never disconnected in our experiments and so is not shown.

\begin{table}[h!]
\centering
\begin{tabular}{ |c|| c c c c c c c c c | }
\hline
 $N$ & $G1$ & $G2$ & $G3$ & $G4$ & $G5$ & $G6$ & $G7$ & $G8$ & $G9$\\ 
 \hline\hline
 $1$ & 28.9 & 80.3 & 0.5 & 0.6 & 0.9 & 0.4 & 1.6 & 6.5 & 1.6 \\
$5$ & 27.6 & 81.5 & 12.4 & 1.1 & 2.1 & 0 & 1.2 & 9.8 & 1.7 \\
$20$ & 27.5 & 79.5 & 11.5 & 1.9 & 3.0 & 0.1 & 2.4 & 15.5 & 2.0 \\
$50$ & 28.5 & 78.8 & 12.2 & 1.1 & 2.7 & 0.1 & 2.4 & 17.1 & 2.6 \\
$100$ & 28.6 & 79.8 & 12.2 & 1.7 & 2.4 & 0.1 & 1.9 & 15.6 & 2.0 \\
\hline
\end{tabular}
\caption{Conditional probabilities (in \%) of nodal RoCoF violations at each generator, given that a RoCoF violation occurs. Results are shown for different time discretizations $N$ of the interval $[0,0.5\mathrm{s}]$.}
\label{table1}
\end{table}
Despite noise due to random sampling, the estimates in Table \ref{table1} are consistent for $N>5$. Recalling \eqref{eq:expsol}, the appropriate choice of $N$ is also informed by the spectral properties of the matrix $A$. In particular, the highest frequency component of the fluctuations is the eigenfunction corresponding to the eigenvalue with largest imaginary part.
Table~\ref{table2} reports some other relevant statistics for the IEEE 39-bus system under the considered disturbance model, namely the probability $p_d$ of multiple RoCoF violations, the average number $\overline{d}$ of violations, and the corresponding average level $\overline{L}$ of lost generation. 

\begin{table}[h!]
\centering
\begin{tabular}{ |c| c c c c c| }
\hline
  & $N=1$ & $N=5$ & $N=20$ & $N=50$ & $N=100$ \\ 
 \hline
 $p_d$ & 15.2\% & 22.4\%  & 24.0\% & 25.0\% & 25.0\% \\
 $\overline{d}$ & 1.21 & 1.37  & 1.44 & 1.46 & 1.44\\
 $\overline{L}$ & 596 & 701 & 735 & 744 & 736\\
\hline
\end{tabular}
\caption{Some statistics for the IEEE 39-bus system: the probability $p_d$ of disconnecting more than one generator, the average number $\overline{d}$ of disconnected generators and the average amount $\overline{L}$ of lost load (in MW).}
\label{table2}
\end{table}

Our case study results highlight the importance of modelling the desynchronization in nodal frequency. It is clear from Table \ref{table1} that the majority of RoCoF violations occur at generator 2, which has a heavy-tailed disturbance model. From Fig.~\ref{fig:ieee39}, generator 2 is connected via a relatively high susceptance line to generator 3, which has a Gaussian disturbance model. Thus RoCoF violations due to a large disturbance at the former generator are capable of inducing subsequent violations at the latter within our considered timescale. This network effect is clearly visible in Fig.~\ref{fig:nodalfrequencies}, where a large initial disturbance at generator 2 is followed by a subsequent RoCoF violation, at around $t=0.4s$, at generator 3. The same relationship can be seen between generators 1 and 8. These observations highlight the importance of the (reduced) network structure and line susceptances in the modelling of frequency violation patterns and system vulnerabilities.


\section{CONCLUSIONS}
\label{sec:conclusions}
This work aims to provide the mathematical framework to understand how unusually large power disturbances cause frequency violations, in particular in terms of RoCoF. We describe the time evolution of the nodal frequencies as a system of coupled swing equations with a random step disturbance at time $0$. A novel MCMC method is introduced, called the ghost sampler, to sample disturbances conditionally on a RoCoF violation occurring, i.e., outside the so-called ``safe region''. An illustrative case study is presented, and it would be of interest to develop this further, for example using empirical probability distributions for heavy-tailed and correlated renewable forecast errors.

Future work will explore further metrics capturing frequency violations, such as the nadir. It would be natural to look also at line overloads caused by power fluctuations and complement in this way the work done in~\cite{Owen2017}. Lastly, we believe that the MCMC ghost sampler has potentially wide applicability beyond power systems reliability in settings where one has to sample rare events. This is particularly so in view of the fact that many of the conditions for the region $K$ can be relaxed.


\section*{ACKNOWLEDGEMENTS}

JM and JV were supported by EPSRC grant EP/P002625/1. AZ is supported by NWO Rubicon grant 680.50.1529. The authors thank Linqi Guo, Janusz Bialek, and Steven H.~Low for helpful discussions on the model.

                                  
\bibliographystyle{alpha}
\bibliography{arxivcdcrocof}

\vspace{1cm}
\section*{APPENDIX}

{\em Proof of Lemma 1:}
For any $r>0$ and $\varphi\in\mathbb{S}^{n-1}$ we have $l^\varphi_x(r)>0$ since $K^c$ is an open set, hence $x$ is the only point that maps to $x$. Then let $T^K_x(x + r_1\varphi_1)=T^K_x(x+r_2 \varphi_2)$ for some $r_1, r_2>0$ and $\varphi_1,\varphi_2\in\mathbb{S}^{n-1}$. This implies $x+l^{\varphi_1}_x(r_1)\varphi_1=x+l^{\varphi_2}_x(r_2)\varphi_2$, so that $\varphi_1 = \varphi_2 = \varphi$ say, and 
$0=l^\varphi_x(r_1)-l^\varphi_x(r_2)=\int_{r_2}^{r_1}\mathbf{1}_{K^c}(x+t\varphi)dt$ must hold, which is only possible if $r_1=r_2$ (again since the set $K^c$ is open).

\hfill $\Box$

{\em Proof of Lemma 2:}
Denote with $S$ the surface of $\mathbb{S}^{n-1}$. Changing \eqref{eq:ghostdens} to polar coordinates and then using the substitution $u:=l^\varphi_x(r)$ (with $du=\mathbf{1}_{K^c}(x+r\varphi)dr$), we have:
\begin{align*}
&Q_K(x,A)=\int_{\RR^n}q_K(x,y)\mathbf{1}_A(y)dy\\
&=\frac{1}{S}\int_{\mathbb{S}^{n-1}}\left(\int_0^{\infty}q_K(x,x+r\varphi) \mathbf{1}_A(x+r\varphi)r^{n-1}dr\right)d\varphi\\
&=\frac{1}{S}\int_{\mathbb{S}^{n-1}}\int_0^{\infty}q(l^{\varphi}_x(r)\varphi) (l^\varphi_x(r))^{n-1}\mathbf{1}_{A \cap K^c}(x+r\varphi)drd\varphi\\
&=\frac{1}{S}\int_{\mathbb{S}^{n-1}}\left(\int_0^{\infty}q(u\varphi) u^{n-1}\mathbf{1}_{T^K_x(A\cap K^c)}(x+u\varphi)du\right)d\varphi\\
&=\int_{\RR^n}q(x,y)\mathbf{1}_{T^K_x(A\cap K^c)}(y)dy=Q(x,T^K_x(A\cap K^c)). \quad \Box
\end{align*}

\begin{defin}\label{def:irred}
A Markov kernel $P$ on a space $S$ is $\nu$-irreducible with respect to a measure $\nu$ (on $S$) if for every $x\in S$ and every measurable $A\subset S$ with $\nu(A)>0$ there exists an $n\in\NN$ such that $P^n(x,A)>0$.
\end{defin}
For the approximation of $\pi$ integrable functions using our MH procedure, a sufficient condition to establish the LLN is the $\pi\mathbf{1}_K$-irreducibility of $Q_K$ on $K^c$ (see, for example \cite[Corollary~2]{tierney} and \cite[Theorem~17.1.7]{tweedie}). Hence, to establish Theorem~\ref{thm:LLN} it is enough to show the following:



\begin{thm}
Let $K$ be closed and ray-bounded. If the underlying proposal density $q$ is strictly positive then the ghost sampling kernel $Q_K$ is $\pi\mathbf{1}_K$-irreducible.
\end{thm}

\begin{proof} Fix an arbitrary $x\in K^c$ and measurable $A\subset \RR^n$ such that $\pi\mathbf{1}_{K^C}(A)=\pi(A\cap K^c)>0$. Define $B=A\cap K^c\cap \{x\in\RR^n;\pi(x)>0\}$. Clearly $\pi(B)>0$. 

Because $K^c$ is open we have $l^{\varphi}_x(r)>0$ for all $r>0$ and all $\varphi\in\mathbb{S}^{n-1}$. By equation \eqref{eq:ghostdens} $q_K$ is strictly positive on $K^c$. Since $B$ must have positive Lebesgue measure we then have
 $Q_K(x,B)>0$ which implies $Q_K(x,A)\geq Q_K(x,B)>0$, so we may take $n=1$ in Definition \ref{def:irred}.
\end{proof}

\addtolength{\textheight}{-12cm} 
                                  
\end{document}